\documentclass[copyright,creativecommons]{eptcs}
\providecommand{\event}{FROM 2023} 

\usepackage{amsmath,epsf,graphicx,amsthm}
\usepackage{iftex}
\newcommand{\keyw}[1]{{\bf #1}}

\newcommand{\dahntab}[1]{
  \newbox\mybok%
  \setbox\mybok=\hbox{\vbox{
      \begin{tabbing}
        #1
      \end{tabbing}%
    }}

  \newdimen\bokwidth%
  \bokwidth=\wd\mybok%
  \newdimen\myl%
  \myl=\textwidth%
  \divide\myl by 2%
  \divide\bokwidth by -2%
  \advance\myl by\bokwidth%
  \vrule width\myl height 0pt depth 0pt%
  \usebox\mybok%
}

\def\e{\varepsilon}
\def\P{\mbox{\it I\hskip -0.25em P}}
\def\E{\mbox{\it I\hskip -0.25em E}}

\def\goesto{\rightarrow}
\newcommand\bigone{\bf 1}

\newtheorem{definition}{Definition}
\newtheorem{theorem}{Theorem}
\newtheorem{lemma}{Lemma}
\newtheorem{corollary}{Corollary}
\newtheorem{proposition}{Proposition}
\def\beginproof{\noindent{\bf Proof.}\quad}
\def\endproof{\qed}

\title{$q$-Overlaps in the Random Exact Cover Problem}
\author{Gabriel Istrate
\institute{University of Bucharest}
\institute{Str. Academiei 14, Sector 6, 011014, \thanks{Corresponding author.}\\
Bucharest, Romania}
\email{gabrielistrate@acm.org}
\and
Romeo Negrea 
\institute{Department of Mathematics, \\ Universitatea Politehnica din Timi\c{s}oara, \\ Victoriei 2, 300006, Timi\c{s}oara, Romania}
\email{romeo.negrea@mat.upt.ro}
}

\newcommand{\titlerunning}{$q$-overlaps in Random Exact Cover}
\newcommand{\authorrunning}{G. Istrate \& R. Negrea}

\begin{document}
\maketitle

\maketitle

\begin{abstract}
We prove lower and upper bounds for the threshold of the following decision problem: given $q\in (0,1)$ and $c>0$ what is the probability that a random instance of the $k$-Exact Cover problem \cite{cs/050837} has two solutions of overlap $qn\pm o(n)$ ? 

These results are motivated by the  {\em one-step replica symmetry breaking} approach of Statistical Physics, and the hope of using an approach based on that of \cite{cond-mat/0504070/prl} to prove that for some values of the order parameter the overlap distribution of $k$-Exact Cover has discontinuous support. 
\end{abstract}


\section{Introduction}

The study of {\em phase transitions in Combinatorial Optimization problems} \cite{sfi-book}, \cite{weigt-hartmann} (see also \cite{istrate:ccc00,aimath04,istrate-sharp,continuous-discontinuous-journal}) has recently motivated (and brought to attention) the geometric structure of the solution space of a combinatorial problem. Methods such as the {\em cavity method} and assumptions such as {\em replica symmetry} and {\em one step replica symmetry breaking} make significant predictions on the geometry of solution space that are a source of inspiration (and a challenge) for rigorous work. 

A remarkable advance in this area is due to M\'{e}zard et al. \cite{cond-mat/0504070/prl}
. This paper has provided rigorous evidence that for the random $k$-satisfiability problem (with sufficiently large $k$) the intuitions concerning the geometry of the solution space provided by the 1-RSB approach are correct. The evidence is based the support of the overlap distribution, shown to be 
discontinuous via a study of threshold properties for the $q$-overlap versions of $k$-SAT. 

In this paper we follow an approach based on the same idea, studying the overlap distribution of a different optimization problem, the {\em random $k$-Exact Cover} problem. The phase transition in 
this problem has been studied in \cite{kalapala-moore-journal}. Zdeborov\'{a} et al. \cite{zdeborova-1},\cite{zdeborova-2} have applied 
nonrigorous methods from Statistical Physics (the cavity approach) and have suggested that the {\em 1-step Replica Symmetry Breaking} assumption is valid. This motivates us to study the problem $q$-overlap $k$-Exact Cover (defined below), and prove lower and upper bounds on its satisfiability threshold. 

Our ultimate goal would be to show that for a certain range of the order parameter the $k$-Exact problem has a discontinuous overlap distribution. However, in this paper we cannot accomplish this goal, as the upper and lower bounds provided are too crude to guarantee this. Still, we believe that the insights provided by our partial result may be useful towards eventually obtaining such bounds. 

\section{Preliminaries} 

We assume knowledge of the method of modeling the trajectory algorithms on random inputs using difference/differential equations using the principle of deferred decision. This is by now a material in standard textbooks \cite{moore2011nature} and surveys \cite{optas-lazy-server}. We will also assume knowledge of somewhat lesser popular techniques in this area, such as the "lazy server" approach \cite{optas-lazy-server}. 

\begin{definition}\label{model} Let ${\cal D} = \{0,1,\ldots, t-1\}$, $t\geq 2$
be a fixed set. Consider the set of all $2^{t^{k}}-1$ potential
nonempty binary constraints on $k$ variables $X_{1}, \ldots, X_{k}$.
We fix a set of constraints ${\cal C}$ and define the random model
$CSP({\cal C})$.  A random formula from $CSP_{n,m}({\cal C})$ is
specified by the following procedure: (i) $n$ is the number of
variables; (ii) we generate uniformly at random, \textbf{with replacement}, $m$ clauses from all the instantiations of constraints in $\mathcal{C}$ on the $n$ variables. 

When all constraints in $\mathcal{C}$ are boolean, we write $SAT({\cal C})$ instead of $CSP({\cal C})$. 
\end{definition}

The particular (CSP) problem  we are dealing with in this paper is: 

\begin{definition} An instance $\Phi$ of the {\em $k$-Exact Cover} is specified by a set of {\em boolean variables} $V=\{x_{1}, \ldots, x_{n}\}$ and a family of $m\geq 1$ subsets of size $k$ (called {\em clauses}) of $V$. Instance $\Phi$ is satisfiable if there is a truth assignment $A$ of variables in $V$ that makes {\em exactly one} variable in each clause evaluate to TRUE. 
\end{definition}

\begin{definition}
The {\em Hamming distance} between two truth assignments $A$ and $B$, on $n$ 
variables is 
$d_{A,B}=\frac{n}{2}-\frac{1}{2}\sum_{i=1}^nA(x_i)B(x_i).$
The {\it overlap} of truth assignments $A$ and $B$ is the fraction of variables on which the two assignments coincide, that is 
$$\mbox{overlap}(A,B)=\frac{\{i|A(x_i)=B(x_i)\}}{n}.$$
\end{definition}

\begin{definition} A set of constraints ${\cal C}$ is {\em interesting}
if there exist constraints  $C_{0},C_{1}\in {\cal C}$ with
$C_{0}(\overline{0})=C_{1}(\overline{1})=0$, where
$\overline{0},\overline{1}$ are the "all zeros" ("all ones")
assignments. Constraint $C_{2}$ is {\em an implicate of $C_{1}$} iff
every satisfying assignment for $C_{1}$ satisfies $C_{2}$.  A
boolean
 constraint $C$ {\em strongly depends on a literal} if it has an unit
 clause as an implicate. A boolean constraint $C$ {\em strongly
 depends on a 2-XOR relation} if $\exists i,j\in \{1,\ldots,k\}$ such
 that constraint ``$x_{i}\neq x_{j}$'' is an implicate of $C$.
\end{definition}

In the following definition $\epsilon(n)$ is a function whose exact expression is unimportant (in that we get the same results), as long as $n^{1/2}=o(\epsilon(n))$, $\epsilon(n)=o(n)$: 

\begin{definition}\label{overlap-model} Let ${\cal D} = \{0,1,\ldots, t-1\}$, $t\geq 2$ be
a fixed set. Let $q$ be a real number in the range [0,1]. The
problem $q$-overlap-$CSP({\cal C})$ is the decision problem
specified as follows: (i) The input is an instance $\Phi$ of
$CSP_{n,p}({\cal C})$; (ii) The decision problem is whether $\Phi$
has two satisfying assignments $A,B$ such that $overlap(A,B)\in[q-\e(n)n^{-1},q+\e(n)n^{-1}]$. 
The random model for $q$-overlap-$CSP({\cal
C})$ is simply the one for $CSP_{n,m}({\cal C})$. 
\end{definition}

This definition particularizes to our problem as follows: 

\begin{definition}
 Let $q\in (0,1)$. The {\it $q$-overlap $k$-Exact Cover} is a decision problem specified as follows: 

{\bf INPUT:} an instance $F$ of k-Exact Cover with $n$ variables. 

{\bf DECIDE:} whether $F$ has two assignments $A$ and $B$ such that 

\begin{equation}\label{ovl}
\mbox{               overlap}(A,B)\in[q-\e(n)n^{-1},q+\e(n)n^{-1}].
\end{equation}

We refer to a pair $(A,B)$ as in equation~(\ref{ovl}) as {\em satisfying assignments of overlap approximately $q$}. 
\end{definition}

If $A,B$ are two satisfying assignments and $i,j\in \{0,1\}$
we will use notation $A=i,B=j$ ($A=B=i$, when $i=j$) as a shorthand for $\{x:A(x)=i,B(x)=j\}$.

\begin{definition} 
 Let $l\geq 1$ be an integer and let $A,B$ be two satisfying assignments of an instance $\Phi$ of $k$ Exact Cover. Pair $(A,B)$ is called {\em $l$-connected} if there exists a sequence of satisfying assignments $A_{0}, A_{1}, \ldots A_{l}$, $A_{0}=A$, $A_{l}=B$, $A_{i}$ and $A_{i+1}$ are at Hamming distance at most $l$. 

\end{definition}
\begin{definition} 
 For $k\geq 3$, $q\in (0,1)$ define 
\begin{equation} 
 q_{k}=\frac{\sqrt{(k-1)(k-2)}}{2+\sqrt{(k-1)(k-2)}}, 
\end{equation} 
and 
\begin{equation}
 \lambda_{q,k}:=\left\{\begin{array}{ll}

\frac{(k-1)q+\sqrt{(k-1)^2q^{2}+k(k-2)(k-1)(1-q)^2}}{2k} & \mbox{ if } q\in (0,q_{k}), \\	 
q & \mbox{ otherwise.} 
\end{array}
\right. 
\label{lqk}
\end{equation} 
\end{definition} 

Note that for $q<q_{k}$ the expression for $\lambda_{q,k}$ is 
the unique positive root of equation 
\begin{equation}\label{l} 
 \frac{k-2}{x}+\frac{(q-2x)}{(k-1)(\frac{1-q}{2})^{2}+x(q-x)}=0,
\end{equation} 
and is strictly less than $q$. Also, $\lambda_{q,k}>q/2$, since, by~(\ref{lqk}),  $\lambda_{q,k}> \frac{(k-1)q}{k}> q/2$. 

\begin{definition} 
 For $k\geq 3$, $q\in (0,1)$ define $F_{k,q}:(q/2,\lambda_{q,k})\goesto (0,\infty)$ by 
\begin{equation} 
 F_{k,q}(x)=\frac{\ln(\frac{x}{q-x})}{\frac{k-2}{x}+\frac{(q-2x)}{(k-1)(\frac{1-q}{2})^{2}+x(q-x)}}
\end{equation} 

\end{definition}

Note that $F_{k,q}$ is well defined, monotonically increasing (the numerator is increasing, each term in the denominator is decreasing), and that 
$\lim_{x\goesto q/2}F_{k,q}(x)= 0$, $\lim_{x\goesto 
\lambda_{q,k}}F_{k,q}(x)=\infty$. Thus function $F_{k,q}$ is a bijection. 
Denote by $G_{k,q}(x)$ its inverse. 

\section{Results} 

We first remark that 

\begin{lemma} 
For every $k\geq 3$ and $q\in (0,1)$ the problem $q$-overlap $k$-Exact Cover has a sharp threshold. 
\end{lemma} 
\begin{proof} 
The claim is a simple application of the main result in \cite{istrate-clustering}. Indeed, in \cite{istrate-clustering} we studied the existence of a sharp threshold for $q$-overlap versions of random constraint satisfaction problems. Previously in \cite{istrate-sharp,creignou-daude-threshold}, a characterization of CSP with a sharp threshold was given:  

\begin{proposition}\label{dichotomy-threshold} Consider a generalized satisfiability problem $SAT({\cal C})$ with
${\cal C}$ interesting. (i) If some constraint in ${\cal C}$
strongly depends on one literal then $SAT({\cal C})$ has a coarse
threshold; (ii) If some constraint in ${\cal C}$ strongly depends on
a 2XOR-relation then $SAT({\cal C})$ has a coarse threshold; (iii)
In all other cases $SAT({\cal C})$ has a sharp threshold.
\end{proposition}

The folowing result (Theorem 8 in \cite{istrate-clustering}) shows that under the same conditions as those in \cite{istrate-sharp} the $q$-overlap versions also have a sharp threshold: 

\begin{proposition}\label{overlap-dichotomy-threshold}
Consider a generalized satisfiability problem $SAT({\cal C})$ such
that (i) ${\cal C}$ is interesting (ii) No constraint in ${\cal C}$
strongly depends on a literal; (iii) No constraint in ${\cal C}$
strongly depends on a 2XOR-relation. Then for all values $q\in
(0,1]$ the problem $q$-overlap-$SAT({\cal C})$ has a sharp
threshold.
\end{proposition}

The conditions in Proposition~\ref{overlap-dichotomy-threshold} apply to the $k$-Exact Cover problem, which can be modeled as a CSP with a single $k$-ary constraint $C_{k}(x_1,x_2,\ldots, x_k)$ which requires that exactly one of $x_1,x_2,\ldots, x_k$ be true. This is because constraint $C_k$ is interesting,  does not strongly depend on a literal and, for $k\geq 3$, does not strongly depend on a 2-XOR relation. 
\end{proof} 

Our main result gives lower and upper bounds on the location of this threshold:  

\begin{theorem} \label{main}
 Let $k\geq 3$ and let $r_{up}(q,k)$ be the smallest $r_{*}>0$ such that $\forall r>r_{*}$
\begin{eqnarray*}
 & & r\ln(P_{k}(G_{k,q}(r),(1-q)/2,(1-q)/2,q-G_{k,q}(r)))-G_{k,q}(r)\ln(G_{k,q}(r))-\\ &-& (q-G_{k,q}(r))\ln(q-G_{k,q}(r))-(1-q)\ln((1-q)/2)\leq 0. 
\end{eqnarray*}
Also let 
\begin{equation} 
 r_{lb}(q)=\left\{\begin{array}{ll}
                   \frac{1}{6}\Big[\frac{1}{(1-q)^2}-1\Big] & \mbox { for } q<1-\frac{1}{\sqrt{2}}, \\
		   \frac{1}{6} & \mbox{ otherwise.}
                  \end{array}
	   \right.
\end{equation} 

Then: 

\begin{itemize} 
 \item[(a).] For $r>r_{up}(q,k)$ a random instance of $q$-overlap $k$-Exact Cover with $n$ variables and $m=rn$ clauses has, with probability $1-o(1)$, no satisfying assignments of overlap approximately $q$. 
\item[(b).] For $0<r<r_{lb}(q)$ a random instance of $q$-overlap $3$-Exact Cover with $n$ variables and $m=rn$ clauses has, with probability $1-o(1)$, two satisfying assignments of overlap approximately $q$. 
\end{itemize}
\label{thm:main}
\end{theorem}

Given the non-explicit nature of $r_{up}(q)$, the only way to interpret the lower and upper bounds given in Theorem~\ref{thm:main} is via symbolic and numeric manipulations of the quantities in the equation(s) defining $r_{up}(q)$. A Mathematica notebook to this goal is provided as \cite{github-overlap-ec}. The conclusion of such an analysis is that the bounds in Theorem~\ref{thm:main} are too crude to imply the existence of a discontinuity in overlap in the $k$-Exact Cover problem. 

\section{Proof of the upper bound (Theorem~\ref{thm:main} (a))}

Let $\Phi$ be a random instance of $k$-Exact Cover. 
Our proof relies on the following fundamental observation: 

\begin{lemma}\label{fundamental}
Let $A,B$ be two satisfying assignments, and let 
$C$ be a clause of length $k$ in $\Phi$. Denote by $c_{0},c_{1},c_{2},c_{3}$ the number of variables of $C$ in the sets $A=B=0$, $A=0,B=1$, $A=1,B=0$, $A=B=1$ respectively. 
Clause $C$ is satisfied by both $A$ and $B$ if and only if
\begin{equation}\label{constraints}
\left\{\begin{array}{lll}
c_{0}=k-2, & c_{1}=c_{2}=1, & c_{3}=0\\
\mbox{ or }\\
c_{0}=k-1, & c_{1}=c_{2}=0, & c_{3}=1
\end{array}\right.
\end{equation}
\end{lemma}

\begin{proof}

The conditions that both $A$ and $B$ satisfy $C$ are written as

\begin{equation} 
\left\{
\begin{array}{ll}
c_{0}+c_{1}=k-1, & c_{2}+c_{3}=1\\
c_{0}+c_{2}=k-1, & c_{1}+c_{3}=1, 
\end{array}\right. 
\end{equation}

a system whose solutions are those from equation~(\ref{constraints}).
\end{proof}

An immediate consequence of Lemma~\ref{fundamental} is that the probability that a pair of assignments satisfies a random instance of $k$-EC depends only on numbers $c_{0},c_{1},c_{2},c_{3}$: 

\begin{lemma} 
 Let $c_{0},c_{1},c_{2},c_{3}$ be nonnegative numbers. Then 
\[
 Pr[A,B\models \Phi\mbox{ }|\mbox { }|A=B=0|=c_{0},\ldots |A=B=1|=c_{3}]=P^{*}(c_{0},c_{1},c_{2},c_{3})^{rn}, 
\]
where 
\begin{equation} 
P^*(a,b,c,d)=
\frac{\Big(\begin{array}{l}a\\ k-2\end{array}\Big)\Big(\begin{array}{l}b\\ 
1\end{array}\Big)\Big(\begin{array}{l}c\\ 1\end{array}\Big)+\Big(\begin{array}{l}a\\ 
k-1\end{array}\Big)\Big(\begin{array}{l}d\\ 1\end{array}\Big)}{\Big(\begin{array}{l}n\\ 
k\end{array}\Big)}=
\frac{\Big(\begin{array}{l}a\\ 
k-2\end{array}\Big)}{\Big(\begin{array}{l}n\\ k\end{array}\Big)}\Big[bc+\frac{(a-k+2)}{(k-1)}d\Big]
\end{equation}
\label{lemma2}
\end{lemma}

\begin{proof} 
We will prove that the probability that $A,B$ satisfy a particular clause of $\Phi$ is $P^{*}(c_0,c_1,c_2,c_3)$. The result follows since the formula $\Phi$ is obtained by sampling independently, with replacement, $rn$ clauses. 

Indeed, the total number of clauses is ${{n}\choose {k}}$. By Lemma~\ref{fundamental}, the number of clauses satisfied by both $A$ and $B$ is ${{a}\choose {k-2}}\cdot {{b}\choose {1}}\cdot {{c}\choose {1}}+$. The first term represents the number of clauses with $k-2$ variables in the set $A=B=0$, one in the set $A=0,B=1$ and one in the set $A=1,B=0$ (so that exactly one literal of $C$ is true in both $A$ and $B$). The second term counts the second type of favorable clauses.  
\end{proof} 

We will use Lemma~\ref{lemma2} to derive an upper bound via the first moment method. 

Indeed, let $Z=Z(q,F)$ be a random variable defined as 

\begin{equation} \label{Z}
Z(q,F)=\sum_{A,B} \delta[|d_{A,B}-nq|\leq e(n)]\cdot  \bigone_{{\mathcal S}(F)}(A)\cdot 
\bigone_{{\mathcal S}(F)}(B).
\end{equation}
where $F=F_k(n,rn)$ is a random formula on $n$ variables over $m=rn$ clauses of size $k$, the 
set ${\mathcal S}(F)$ is the set of the EC-assignments to this formula. 

Then: 

\begin{equation} \label{EZ}
E[Z(q,F)]=\sum_{A,B} \delta[|d_{A,B}-nq|\leq e(n)]\cdot  Pr[A,B\models F].
\end{equation}

For fixed values $a,b,c,d$ there are $\binom{n}{a, b, c, d}= \frac{n!}{a!\cdot b!\cdot c!\cdot d!}$ pairs of assignments of type $(a,b,c,d)$. If we denote $\lambda\stackrel{not}{=}a+d=nq\pm\e(n)$ and $\mu\stackrel{not}{=}b+c=n-\lambda$ 
then the system
$$\left\{\begin{array}{l}a+d=\lambda\\b+c=n-\lambda\end{array}\right.$$
has at most $(\lambda+1)(n-\lambda+1)$ solutions in the set of nonnegative integers. Therefore, the number of quadruples $(a,b,c,d)$ in the sum $E[Z]$ is at most 
$$\sum_{\lambda=nq-\e(n)}^{nq+\e(n)}(\lambda+1)(n-\lambda+1)=
\frac{1}{3}(1+2\e(n))(3-\e(n)-\e(n)^2+3n+3n^2q-3n^2q^2)\stackrel{def}{=}M.$$

So 
\begin{equation} 
 P[Z>0]\leq E[Z]\leq M\cdot \max_{(a,b,c,d)}\binom{n}{a, b, c, d}\cdot  P^{*}(a,b,c,d)^{rn}
\end{equation} 
We will compute the maximum on the right-hand side and derive conditions for which this right-hand side tends (as $n\rightarrow \infty$) to zero. 

Indeed, denote $\alpha = \frac{a}{n}, \beta=\frac{b}{n}, \gamma=\frac{c}{n},\delta=\frac{d}{n}$. Applying 
Stirling's formula $n!=(1+o(1))\cdot \big(\frac{n}{e}\big)^{n}\sqrt {2\pi n}$, and also noting that 
\[P^{*}(a,b,c,d)\leq (1+\frac{O(1)}{n})\cdot P_k(\alpha,\beta,\gamma,\delta),
\]
with
\begin{equation} 
P_k(\alpha,\beta,\gamma,\delta)=
\alpha^{k-2}k(k-1)(\beta\gamma +\frac{\alpha\delta}{k-1})
\end{equation} 
we get 
\[
 P[Z>0]\leq M\cdot \theta(1)\cdot  \max_{(\alpha,\beta,\gamma,\delta)}\cdot\big[\Big(\frac{1}{\alpha^{\alpha}\beta^{\beta}\gamma^{\gamma}\delta^{\delta}}\Big)\cdot P(\alpha,\beta,\gamma,\delta)^{r}\big]^n
\]
Define 
\[
 g_{r}(\alpha,\beta,\gamma,\delta)=\frac{P_k(\alpha,\beta,\gamma,\delta)^{r}}{\alpha^{\alpha}\beta^{\beta}\gamma^{\gamma}\delta^{\delta}}
\]
\begin{lemma} 
 For any $r>0$ we have 
\[
 \max\Big\{g_{r}(\alpha,\beta,\gamma,\delta):\alpha+\delta=q, \beta+\gamma=1-q, \alpha,\beta,\gamma,\delta\geq 0\Big\}= g_{r}(\alpha_{*,r},\beta_{*,r},\gamma_{*,r},\delta_{*,r}),
\]
with 
\begin{equation} \label{max}
\left\{
\begin{array}{l}
\alpha_{*,r}=G_{k,q}(r), \\ 
 \beta_{*,r}=\gamma_{*,r}=(1-q)/2,\\
\delta_{*,r}=q-G_{k,q}(r). 
\end{array}\right. 
\end{equation}
\end{lemma}
\beginproof 

First, it is easy to see that 
\begin{equation}\label{beta-gamma}
 g_{r}(\alpha, \beta,\gamma,\delta)\leq g_{r}\Big(\alpha, \beta_{*,r},\gamma_{*,r},\delta\Big).
\end{equation} 
Indeed, function $x\ln(x)$ is convex, having the second derivative positive, and $e^{x}$ is increasing so, by Jensen's inequality, 
\[
\beta^{\beta}\gamma^{\gamma}=e^{\beta\ln(\beta)+\gamma\ln(\gamma)}\geq e^{(\beta+\gamma)\ln(\frac{\beta+\gamma}{2})}=\Big(\frac{\beta+\gamma}{2}\Big)^{\beta+\gamma}= \beta_{*}^{\beta_{*,r}}\gamma_{*,r}^{\gamma_{*,r}}. 
\]
On the other hand since $\beta\gamma \leq \Big(\frac{\beta+\gamma}{2}\Big)^2=\beta_{*,r}\gamma_{*,r}$, we have  $P(\alpha,\beta,\gamma,\delta)\leq P\Big(\alpha, \beta_{*,r},\gamma_{*,r},\delta\Big)$ 
and equation~(\ref{beta-gamma}) follows. 

Also
\begin{equation}\label{alpha-beta}
g_{r}\Big(\alpha, \beta_{*,r},\gamma_{*,r},\delta\Big)\leq 
g_{r}\Big(\alpha_{*,r}, \beta_{*,r},\gamma_{*,r},\delta_{*,r}\Big)
\end{equation} 

Indeed, replacing $\delta = q - \alpha$, the expression 
\begin{eqnarray*}
t(\alpha) & = & \ln g_{r}(\alpha,\beta_{*,r},\gamma_{*,r},q-\alpha)= \\ & = & r\ln\Big(P_k(\alpha,\beta_{*,r},\gamma_{*,r},q-\alpha)\Big)-\alpha\ln(\alpha)-(q-\alpha)\ln(q-\alpha)-\beta_{*,r}\ln(\beta_{*,r})-\gamma_{*,r}\ln(\gamma_{*,r})
\end{eqnarray*}
 is a function of $\alpha$ whose derivative is 
\begin{eqnarray*}
 t^\prime{(\alpha)}& = & r \frac{P^{\prime}_{k}(\alpha,\beta_{*,r},\gamma_{*,r},q-\alpha)}{P_k(\alpha,\beta_{*r},\gamma_{*r},q-\alpha)}-\ln(\alpha)-1+\ln(q-\alpha)+1 = \\
 & = & r [\frac{k-2}{\alpha}+\frac{q-2\alpha}{\Big(\frac{1-q}{2}\Big)^2+\alpha(q-\alpha)}]+\ln(\frac{q-\alpha}{\alpha}). 
\end{eqnarray*}

so $t(\alpha)$ has a maximum on $[0,q]$ at $\alpha_{*,r}$ which is a solution of equation  
\begin{equation}
 \frac{r(k-2)}{\alpha}+\frac{r(q-2\alpha)}{(k-1)\Big(\frac{1-q}{2}\Big)^2+\alpha(q-\alpha)}=\ln(\frac{\alpha}{q-\alpha}),
 \label{eq:tprime} 
\end{equation} 
or $F_{k,q}(\alpha_{*,r})=r$. In other words $\alpha_{*,r}=G_{k,q}(r)$ and $\delta_{*,r}=q-G_{k,q}(r)$. 
\endproof 
\vspace{5mm}

Formula~(\ref{max}) implies that $P[Z>0]\stackrel{\small n\goesto \infty}{\goesto} 0$ as long as $t(\alpha_{*,r})<0$. 
The critical value $r_{up}(q,k)$ is 
therefore given by equation $t(\alpha_{*,r})=0$, that is

\begin{equation} 
\begin{aligned}
 & & r\ln(P_{k}(G_{k,q}(r),(1-q)/2,(1-q)/2,q-G_{k,q}(r)))-G_{k,q}(r)\ln(G_{k,q}(r))-\\ &-& (q-G_{k,q}(r))\ln(q-G_{k,q}(r))-(1-q)\ln((1-q)/2)= 0. 
\end{aligned}
\label{eq:t}
\end{equation} 
\endproof

For $k=3$, denoting (for simplicity) $
\alpha=G_{3,q}(r)$, we have $P_{3}(\alpha,\beta,\gamma,\delta)=6\alpha(\beta\gamma+\frac{\alpha\delta}{2})$, so the equation~(\ref{eq:t}) becomes
\begin{align*} 
r\ln(6\alpha[(\frac{1-q}{2})^2+\frac{\alpha(q-\alpha)}{2}])-\alpha\ln(\alpha)-(q-\alpha)\ln(q-\alpha)=(1-q)\ln((1-q)/2)
\end{align*} 
while equation~(\ref{eq:tprime}) becomes
\[
r[1+\frac{\alpha(q-2\alpha)}{2\Big(\frac{1-q}{2}\Big)^2+\alpha(q-\alpha)}]=\alpha \ln(\frac{\alpha}{q-\alpha}), 
\]
Attempting a substitution of the type $\frac{\alpha}{q-\alpha}=t$ in this last equation seems to turn the function $F_3$ into a generalized version of the Lambert function. However, this generalization seems to be different from the versions already existing in the literature \cite{mezo2022lambert}, so this attempt does not seem fruitful.  

We refer again to the Mathematica notebook provided as \cite{github-overlap-ec}. In particular let us remark that the maximum value of $r_{up}(q)$ is reasonably close to upper bound the threshold for $3$-Exact Cover derived using the first-moment method in \cite{knysh-2004}.

\section{Proof of the lower bound (Theorem~\ref{thm:main} (b))}

We will use a constructive method. Just as in \cite{cs/050837}, we will derive a lower bound from the probabilistic analysis of an algorithm. However, the algorithm {\bf will not} be the one from \cite{cs/050837}. Instead, we will investigate (a variant of) the algorithm LARGEST-CLAUSE in Figure~\ref{alg}.

\begin{figure}
\fbox{
\dahntab{
\ \ \ \ \=\ \ \ \ \=\ \ \ \ \=\ \ \ \ \= \ \ \ \ \=\ \ \ \ \= \\
Algorithm {\bf LargestClause} \\
\\
{\bf INPUT}: a formula $\Phi$ \\
\\
if ($\Phi$ contains a unit clause)\\
\> choose a random unit clause ${l}$\\
\> set $l$ to TRUE\\
\> \> if this creates a contradiction FAIL \\
\> \> else call the algorithm recursively \\
else if ($\Phi$ contains a clause of length $\geq 3$)\\
\> \> choose a random clause $C$ of maximal length\\
\> \> choose a random literal $l$ of $C$ \\
\> \> set $l$ to zero and simplify the formula\\
\> else \\
\> \> \> create a graph $G$ containing an edge $(x,y)$\\
\> \> \> for any clause $x\oplus y$ in $\Phi$; \\
\> \> \> if ($G$ is not bipartite) OR (some connected component has size $\geq f(n)$)\\
\> \> \> \> FAIL\\
\> \> \> else \\
\> \> \> \> choose one variable in each connected component of G\\
\> \> \> \> create satisfying assignments $A$ and $B$ \\
\> \> \> \> by setting all chosen variables to one (zero) \\
\> \> \> \> and then propagating these values to all variables in $G$. \\
\> \> \> \> \keyw{return} $(A,B)$.   
}
}
\caption{Algorithm LARGEST-CLAUSE}\label{alg}
\end{figure}

Intuitively, the reason we prefer the algorithm LARGEST-CLAUSE to the one from \cite{cs/050837} is simple: unlike \cite{cs/050837}, our goal is not to simply solve an instance of $k$-EXACT COVER, but to create \textbf{two satisfying assignments} of controlled overlap. We we would like to accomplish that via an algorithm that iteratively assigns values to variables and is left (at some point) with solving a 2-XOR SAT formula. Our aim is to keep the number of set variables to a minimum, in order to create satisfying assignments with as large an overlap as possible. But that means that one must ``destroy'' all clauses of length different from two as fast as possible. Instead, the algorithm in \cite{cs/050837} is focused on killing clauses of length 2. 

The algorithm may seem incompletely specified, as its performance depends on function $f(n)$. As it will turn out, the precise specification of function $f(n)$ in the algorithm LARGEST-CLAUSE will not matter for our purposes, as long as $f$ is a function that grows asymptotically faster than the size of the giant component in a certain subcritical Erd\H{o}s-R\'enyi random graph, which is with high probability $O(log(n))$.

To analyze (versions of) Algorithm LARGEST-CLAUSE, we denote by $C_{i}(t)$, $i\geq 2$, the number of clauses of length $i$ that are present after $t$ variables have been set. Also define
$P(t),N_{t}$ to be the number of positive (negative) unit clauses present at stage $t$. Finally, define
 functions $c_{1},c_{2},c_{3},p,n:(0,1)\rightarrow {\bf R}_{+}$ by $
 c_{i}(\alpha)=C_{i}(\alpha \cdot n)/n$, 
and similar relations for functions $p(\cdot),n(\cdot)$. 
We will use a standard method, {\em the principle of deferred decisions} to analyze algorithm LARGEST-CLAUSE. See \cite{optas-lazy-server} for a tutorial. 

It is easy to show by induction that at any stage $t$, conditional on the four-tuple $(P(t),N(t),C_{2}(t),C_{3}(t))$, the remaining formula is uniform. 

We divide the algorithm in two phases: in the {\em first phase} there exist clauses of length three. In the {\em second phase} only clauses of length one and two exist. 

If a variable is set to TRUE then a 
1-in-$i$ clause containing that variable is turned into $i-1$ negative unit clauses. If a variable is set to FALSE then a 
1-in-$i$ clause is turned into a 1-in-$(i-1)$ clause, in particular a 1-in-2 clause is turned into a positive unit clasue. The dynamics is displayed in Figure~\ref{dyn}. 

\begin{figure}\label{dyn}
\begin{center} 
\includegraphics[width=8cm]{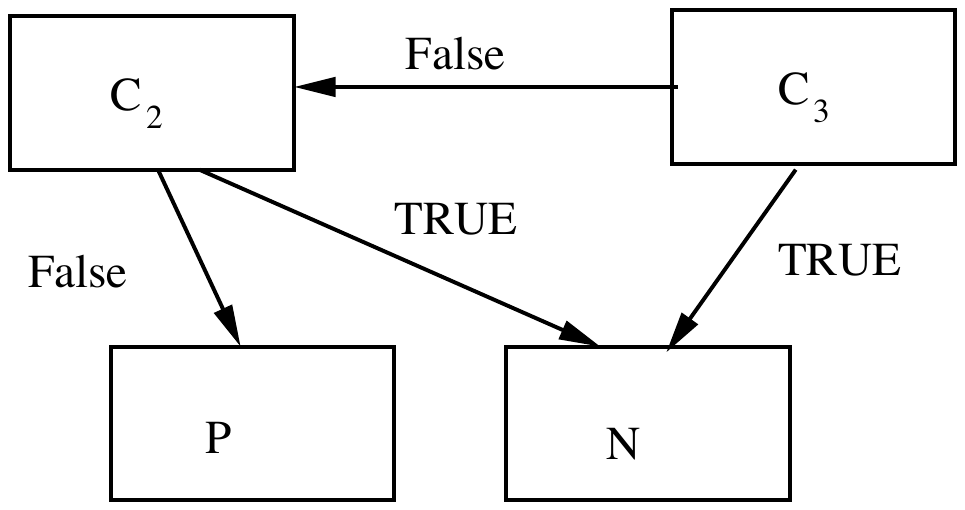}
\end{center} 
\caption{The dynamics of algorithm LARGEST-CLAUSE.}
\end{figure}

The different dynamics of the flows in the cases when a positive (negative) literal is set makes the direct analysis of algorithm LARGEST-CLAUSE difficult. Therefore, we will instead analyze a 
version of the algorithm, given in Figure~\ref{lazy-alg}, using a ``lazy-server'' \cite{optas-lazy-server} idea. Specifically, instead of always trying to simplify the unit clauses, we will do so probabilistically (see Figure~\ref{lazy-alg} for details). 

Since the problems $q$-overlap EXACT COVER have a sharp threshold, it is enough to prove, for $r<r_{lb}(q)$ that the algorithm finds a 
pair of assignments of overlap $r$ with probability $\Omega(1)$. This will be enough to conclude that with probability $1-o(1)$ two satisfying assignments of overlap $r$ exist. 

\begin{figure}
\fbox{
\dahntab{
\ \ \ \ \=\ \ \ \ \=\ \ \ \ \=\ \ \ \ \= \ \ \ \ \=\ \ \ \ \= \\
Algorithm {\bf LazyLargestClause} \\
\\
{\bf INPUT}: a formula $\Phi$ \\
\\
{\bf if} (at least one of the alternatives 1,2,3 below applies) \\
\> take one of the following actions with probabilities  $\lambda_{1}(t),\lambda_{2}(t),\lambda_{3}(t)$, respectively: \\
\> 1. if ($\Phi$ contains a positive unit clause)\\
\> \> choose a random positive unit clause ${l}$\\
\> \> set $l$ to TRUE\\
\> \> \> if this creates a contradiction FAIL \\
\> \> else \\
\> \> \> set a random variable to TRUE \\
\> 2. if ($\Phi$ contains a negative unit clause)\\
\> \> choose a random negative unit clause $\overline{l}$\\
\> \> set $l$ to FALSE\\
\> \> \> if this creates a contradiction FAIL \\
\> \> else \\ 
\> \> \> set a random variable to FALSE \\
\> 3. if ($\Phi$ contains a clause of length $\geq 3$)\\
\> \> choose a random clause $C$ of maximal length\\
\> \> choose a random literal $l$ of $C$ \\
\> \> set $l$ to FALSE\\
{\bf else} \\
\> \> \> run the corresponding bipartite graph construction \\
\> \> \> of algorithm LARGEST-CLAUSE. \\ 
}
}
\caption{The ``lazy-server'' version of algorithm LARGEST-CLAUSE}\label{lazy-alg}
\end{figure}

Let $U_{P}(t),U_{N}(t),U_{3}(t)$ be 0/1 variables that are one exactly when choice 1 (2,3) is selected, 0 otherwise. We can write the following recurrence relations describing the dynamics of the four-tuple $(P(t),N(t),C_{2}(t),C_{3}(t))$:
\begin{equation} 
\left\{ \begin{array}{l}
C_{3}(t+1)=C_{3}(t)-U_{3}(t)-\Delta_{3}(t),\\
C_{2}(t+1)=C_{2}(t)-\Delta_{2}(t)+\Delta_{3,2}(t),\\
P(t+1)=P(t)-U_{P}(t)-\Delta_{1,P}(t)+\Delta_{2,P}(t),\\
N(t+1)=N(t)-U_{N}(t)-\Delta_{1,N}(t)+\Delta_{2,N}(t)+\Delta_{3,N}(t),
\end{array}
\right.
\end{equation} 
where 
\begin{equation} 
 \left\{ \begin{array}{l}
\Delta_{3}(t)\stackrel{d}{=}Bin\Big(C_{3}(t)-U_{3}(t),\frac{3}{n-t}\Big). \\  
\Delta_{2}(t)=\Delta_{2,N}(t)+\Delta_{2,P}(t)\stackrel{d}{=}Bin\Big(C_{2}(t),\frac{2}{n-t}\Big). \\
\Delta_{3,2}(t)\stackrel{d}{=}U_{3}(t)+(U_{N}(t)+U_{3}(t))\cdot Bin\Big(C_{3}(t)-U_{3}(t),\frac{3}{n-t}\Big)\\
\Delta_{3,N}(t)\stackrel{d}{=}2U_{P}(t)\cdot Bin\Big(C_{3}(t),\frac{3}{n-t}\Big)\\
\Delta_{2,P}(t)\stackrel{d}{=}(U_{N}(t)+U_{3}(t))\cdot Bin\Big(C_{2}(t),\frac{2}{n-t}\Big)\\
\Delta_{2,N}(t)\stackrel{d}{=}U_{P}(t)\cdot Bin\Big(C_{2}(t),\frac{2}{n-t}\Big)\\
\Delta_{1,P}(t)\stackrel{d}{=}Bin\Big(P(t)-U_{P}(t),\frac{1}{n-t}\Big)\\
\Delta_{1,N}(t)\stackrel{d}{=}Bin\Big(N(t)-U_{N}(t),\frac{1}{n-t}\Big)\\
         \end{array}
\right.
\end{equation}

By an analysis completely similar to that of algorithm for random $k$-SAT (see e.g. \cite{optas-lazy-server}), we derive the following system of equations that describe the average trajectory path of Algorithm LAZY LARGEST-CLAUSE: 
\begin{equation}\label{diffeq}
\left\{\begin{array}{l}
c_{3}^{\prime}(t)=-\lambda_{3}(t)-\frac{3c_{3}(t)}{(1-t)}.\\
c_{2}^{\prime}(t)=-\frac{2c_{2}(t)}{(1-t)}+\frac{3c_{3}(t)}{(1-t)}\cdot (\lambda_{2}(t)+\lambda_{3}(t)),\\
\end{array}
\right.
\end{equation}

with initial conditions $(c_{2}(0),c_{3}(0))=(0,r)$.

In this paper we will make the simplest choice
\begin{equation} \lambda_{1}(t)=\lambda_{2}(t)=\lambda_{3}(t)=1/3.
\end{equation}

Differential equations~(\ref{diffeq}) describe the dynamics of algorithm LARGEST-CLAUSE only for $t\in [t_{3},t_{2})$, where $t_{3}=0$ and $t_{2}\in (0,1)$ is the smallest solution of equation $c_{3}(t)=0$.  

Simple computations lead us to formulas: 
\begin{equation}\label{first-stage}
 \left\{\begin{array}{l}
c_{3}(t)=(r+\frac{1}{6})(1-t)^{3}-\frac{1-t}{6},\\
c_{2}(t)=\frac{(1-t)^2}{3}-\frac{(1-t)}{3}+2(r+\frac{1}{6})t(1-t)^{2},\\
        \end{array}
\right.
\end{equation}

which describe the dynamics of algorithm LARGEST-CLAUSE in range $0\leq t < t_{2}=1-\frac{1}{\sqrt{6r+1}}$.

The average flow into positive unit clauses is 
\begin{eqnarray*}
 F^{P}_{2}(t) &:= & \frac{2}{3}\cdot \frac{2c_{2}(t)}{1-t}+\frac{1}{3}\cdot\frac{2\cdot 3c_{3}(t)}{1-t}= \\ & = & \frac{4}{3} \Big[\frac{(1-t)^2}{3}-\frac{(1-t)}{3}+2(r+\frac{1}{6})t(1-t)\Big]+2(r+\frac{1}{6})(1-t)^{2}-\frac{1}{3}.
\end{eqnarray*}
\[
(F^{P}_{2})^{\prime}(t)=\frac{8r(1-2t)}{3}-4(r+\frac{1}{6})(1-t)=\Big(\frac{4r}{3}-\frac{2}{3}\Big)(1-t)-\frac{8r}{3}<0, 
\]
so $F^{P}_{2}(t)$ has a 
maximum at $0$, equal to $2r$. For $r<1/6$ this is less than $1/3$, so it is balanced by being given the opportunity (with probability 1/3) to consume a positive unit clause, if any. 
 
The average flow into negative unit clauses is  
\[
 F^{N}_{2}(t) =\frac{1}{3}\cdot\frac{2c_{2}(t)}{1-t}=\frac{2}{3}\cdot \Big[\frac{(1-t)}{3}-\frac{1}{3}+2(r+\frac{1}{6})t(1-t).\Big]=\frac{2t}{9}\Big[(6r+1)(1-t)-1\Big]. 
\]

The maximum of $F^{N}_{2}(t)$ is reached at $t=\frac{3r}{6r+1}$, which is in the interval $(t_{3},t_{2})$
for $r> 0$, and is equal to $\frac{2r^2}{6r+1}=\frac{r}{3}(1-\frac{1}{6r+1})$, which is definitely less than $\frac{1}{3}$, for $r<1/6$. 

The conclusion is that for $r<1/6$  with probability $1-o(1)$ both flows into positive and negative unit clauses can be handled by the lazy server with choice $\lambda_{1}=\lambda_{2}=\lambda_{3}=1/3$ without creating contradictory clauses. 

Around stage $t_{2}n\pm o(n)$ clauses of length three and one run out. We are left with a system of $(c_{2}(t_{2})+o(1))n$
1-in-2 clauses in the remaining $\overline{n}=(1-t_{2})n$ variables. 
Consider graph $G$ corresponding to these equations, where for every equation $x\oplus y=1$ we add edge $(x,y)$ to $G$. 

By the uniformity lemma {\em 
$G$ can be seen as an Erd\H{o}s-Renyi random graph $G(\overline{n},\frac{\mu}{\overline{n}})$}, with  probability coefficient 
\[
\mu =2c_{2}(t_{2})/(1-t_{2}) =3F_{2}(t_{2}). 
\]
  
Our maximum computation shows that for $r\in (0,1/6)$, $3F_{2}(t_{2})<1$. Thus $G$ is a subcritical random graph, whose connected components are w.h.p. of size $O(\log n)$. With constant probability (depending only on $\mu$), $G$ is a bipartite graph. In this situation giving a value to an arbitrary node uniquely determines the values of variables in the connected component. 

We create two assignments $A$ and $B$ as follows: 
\begin{enumerate} 
 \item On variables $x$ set by algorithm LARGEST-CLAUSE, $A(x)=B(x)$, equal to the value given by the algorithm.  
\item On variables in graph $G$ $A$ and $B$ take opposite values. This can be accomplished by giving $A,B$ different values on a set of fixed variables, one in each connected component of $G$. 
\end{enumerate}

When graph $G$ is bipartite $A$ and $B$ are satisfying assignments. When the connected components of $G$ are of size $O(\log n)$ we can create a path from $A$ to $B$ consisting satisfying assignments by consecutively flipping values of variables on which $A$ and $B$ are different, one connected component at a time. The overlap of $A$ and $B$ is equal to $1-\frac{1}{\sqrt{6r+1}}$. 

It follows that for any $q\in (0,1)$, the $q$-overlap Exact Cover is satisfiable w.h.p. for $q>1-\frac{1}{\sqrt{6r+1}}$, i.e. $\frac{1}{6r+1}>(1-q)^2$, which can be rewritten as  $r<r_{lb}(q)$. 
\endproof

\section{Remarks} 

The condition $r<1/6$ in Theorem~\ref{main} has an easy probabilistic interpretation: it is the location of the phase transition for the random 3-uniform hypergraph \cite{phase-transition-random-hypergraph}. In this range most connected components are small and tree-like or unicyclic, so the space of variables breaks down in independent clusters of size $O(\log n)$. Thus we should expect that all overlaps in 
some range $(\lambda,1)$ are satisfied with probability $1-o(1)$, which is exactly what happens, according to Theorem~\ref{main}, for $\lambda = 1-\frac{1}{\sqrt{2}}$.

In fact we can state more: in this regime there is a single cluster of solutions, and the bounds on the overlap we provide are in fact bounds on the diameter of this cluster. 

\begin{theorem}
 Let $r<1/k(k-1)$. There exists $C>0$ such that, with probability $1-o(1)$ (as $n\goesto \infty$), if $\Phi$ is a random instance of $k$-Exact-Cover with $n$ variables and $rn$ clauses, any two satisfying assignments of $\Phi$ are $C\log(n)$ connected. 
\end{theorem}

\beginproof
Since the formula hypergraph $H$ of $\Phi$ is subcritical, there exists \cite{phase-transition-random-hypergraph} $C>0$ such that w.h.p. all connected components of $H$ have size at most $C\log(n)$. That means that formula $\Phi$ is the decomposition of several variable-disjoint formulas $\Phi_{1},\ldots, \Phi_{p}$. In turn, satisfying assignments for $\Phi$ are obtained by concatenating satisfying assignments for these formulas. 

This argument immediately implies that any two satisfying assignments of $\Phi$ are $C\log(n)$ connected: Let $A,B$ be two such satisfying assignments, and let $A_{1},B_{1}$, $(A_{2},B_{2})$, $\ldots, (A_{s},B_{s})$ be the restrictions of $A$ and $B$ on the components on which they differ. 

One can obtain a path from $A$ to $B$ as follows (where variables $x$ such that $A(x)=B(x)$ are ommitted from representation):  
\begin{eqnarray*}
 A & = &(A_{1},A_{2},\ldots, A_{s})\goesto (B_{1},A_{2},\ldots, A_{s})\goesto (B_{1},B_{2},\ldots, A_{s})\goesto \ldots \goesto \\ 
&\goesto & (B_{1}, \ldots, B_{s-1},A_{s})\goesto (B_{1},B_{2},\ldots B_{s})=B. 
\end{eqnarray*}

The intermediate assignments are satisfying assignments since formulas $\Phi_{1}, \ldots$, $\Phi_{p}$ are disjoint. They are 
at distance at most $C\log(n)$ because of the upper bound on the component size of $H$. 

\endproof

Using the above result we obtain the following analog of the result proven in \cite{istrate-clustering} for 2-SAT: 

\begin{corollary} 
 For $r<1/k(k-1)$ a random instance of $k$-Exact Cover has a single cluster of satisfying assignments and an overlap distribution with continuous support. 
\end{corollary}

The relative weakness of the bound $r<1/k(k-1)$ comes from our suboptimal choice of parameters $\lambda_{1}(t),\lambda_{2}(t),\lambda_{3}(t)$. For instance, for $k=3$ the bound $r<1/6$ comes entirely from handling positive unit clauses, while we have no problem satisfying negative ones, since the flow $F^{N}_{2}(t)$ always stays below one. This suggests that we are disproportionately often taking care of negative unit literals.

In what follows we sketch an approach for a better choice of these parameters. We were not able to explicitly calculate $\lambda_{1}(t),\lambda_{2}(t),\lambda_{3}(t)$, so we are unable to offer an improved analysis of the LAZY LARGEST-CLAUSE algorithm. 

First, the algorithm has to be able to satisfy the positive unit flow, so 
\[
 \lambda_{1}(t)\geq (\lambda_{2}(t)+\lambda_{3}(t))\cdot \frac{2c_{2}(t)}{1-t}. 
\]
Thus 
\[
 \frac{\lambda_{1}(t)}{1-\lambda_{1}(t)}\geq \frac{2c_{2}(t)}{1-t}
\]
in other words 
\[
 \lambda_{1}(t)\geq \frac{2c_{2}(t)}{1-t+2c_{2}(t)}. 
\]
First, the algorithm has to be able to handle the negative unit flow, so 
\[
 \lambda_{2}(t)\geq \lambda_{1}(t)\frac{6c_{3}(t)+2c_{2}(t)}{1-t}
\]
We choose 
\begin{equation} 
\left\{\begin{array}{l}
        \lambda_{1}(t)=\frac{2c_{2}(t)+\epsilon}{1-t+2c_{2}(t)}\\
	\lambda_{2}(t)=\frac{(2c_{2}(t)+\epsilon)(6c_{3}(t)+2c_{2}(t)+\epsilon)}{(1-t)(1-t+2c_{2}(t))},\\
	\lambda_{3}(t)=1-\lambda_{1}(t)-\lambda_{2}(t)=\frac{(2c_{2}(t)+\epsilon)(6c_{3}(t)+2c_{2}(t)+\epsilon)}{(1-t)(1-t+2c_{2}(t))}. 
       \end{array}
\right.
\end{equation}

It is an open problem if this approach can be completed to a full analysis. 

\section{Conclusions}

The obvious question raised by this work is to improve our bounds enough to display the discontinuity of overlap distribution, a property of $k$-Exact Cover we believe to be true. 

Note that there are obvious candidate approaches to improving our bounds: first, the lower bound could be improved by trying a rigorous version of the (heuristic) upper bound approach of Knysh et al. 
\cite{knysh-2004}. Or, it could be improved by finding explicit expressions for the parameters in (and explicitly analyzing) the LAZY LARGEST-CLAUSE algorithm, along the lines described in the previous section. Neither one of these two approaches looks particularly tractable, though. 

As for the upper bound, an obvious candidate is the second moment method. We have attempted such an approach. The problem is that it seems to require optimizing of a function of 16 variables without enough obvious symmetries that would make the problem tractable. 

\section{Acknowledgments} 

The authors thank the anonymous referees for useful comments, suggestions and corrections.  

\bibliographystyle{eptcs}

\begin{thebibliography}{10}
  \providecommand{\bibitemdeclare}[2]{}
  \providecommand{\surnamestart}{}
  \providecommand{\surnameend}{}
  \providecommand{\urlprefix}{Available at }
  \providecommand{\url}[1]{\texttt{#1}}
  \providecommand{\href}[2]{\texttt{#2}}
  \providecommand{\urlalt}[2]{\href{#1}{#2}}
  \providecommand{\doi}[1]{doi:\urlalt{https://doi.org/#1}{#1}}
  \providecommand{\eprint}[1]{arXiv:\urlalt{https://arxiv.org/abs/#1}{#1}}
  \providecommand{\bibinfo}[2]{#2}
  
  \bibitemdeclare{article}{optas-lazy-server}
  \bibitem{optas-lazy-server}
  \bibinfo{author}{D.~\surnamestart Achlioptas\surnameend}
    (\bibinfo{year}{2001}): \emph{\bibinfo{title}{Lower bounds for random 3-{SAT}
    via differential equations}}.
  \newblock {\slshape \bibinfo{journal}{Theoretical Computer Science}}
    \bibinfo{volume}{265}, pp. \bibinfo{pages}{159--185},
    \doi{10.1016/S0304-3975(01)00159-1}.
  
  \bibitemdeclare{inproceedings}{creignou-daude-threshold}
  \bibitem{creignou-daude-threshold}
  \bibinfo{author}{N.~\surnamestart Creignou\surnameend} \&
    \bibinfo{author}{H.~\surnamestart Daud\'{e}\surnameend}
    (\bibinfo{year}{2004}): \emph{\bibinfo{title}{Coarse and Sharp Thresholds for
    Random Generalized Satisfiability Problems}}.
  \newblock In \bibinfo{editor}{M.~Drmota \surnamestart et~al.\surnameend},
    editor: {\slshape \bibinfo{booktitle}{Mathematics and Computer Science III:
    Algorithms, Trees, Combinatorics and Probabilities}},
    \bibinfo{publisher}{Birkhauser}, pp. \bibinfo{pages}{507--517},
    \doi{10.1007/978-3-0348-7915-6}.
  
  \bibitemdeclare{book}{weigt-hartmann}
  \bibitem{weigt-hartmann}
  \bibinfo{author}{A.~\surnamestart Hartmann\surnameend} \&
    \bibinfo{author}{M.~\surnamestart Weigt\surnameend} (\bibinfo{year}{2005}):
    \emph{\bibinfo{title}{Phase transitions in combinatorial optimization
    problems}}.
  \newblock \bibinfo{publisher}{Wiley-VCH}, \doi{10.1002/3527606734}.
  
  \bibitemdeclare{inproceedings}{istrate:ccc00}
  \bibitem{istrate:ccc00}
  \bibinfo{author}{G.~\surnamestart Istrate\surnameend} (\bibinfo{year}{2000}):
    \emph{\bibinfo{title}{Computational Complexity and Phase Transitions}}.
  \newblock In: {\slshape \bibinfo{booktitle}{Proceedings of the 15th I.E.E.E.
    Annual Conference on Computational Complexity (CCC'00)}}, pp.
    \bibinfo{pages}{104--115}, \doi{10.1109/CCC.2000.856740}.
  
  \bibitemdeclare{article}{istrate-sharp}
  \bibitem{istrate-sharp}
  \bibinfo{author}{G.~\surnamestart Istrate\surnameend} (\bibinfo{year}{2005}):
    \emph{\bibinfo{title}{Threshold properties of random boolean constraint
    satisfaction problems}}.
  \newblock {\slshape \bibinfo{journal}{Discrete Applied Mathematics}}
    \bibinfo{volume}{153}, pp. \bibinfo{pages}{141--152},
    \doi{10.1016/j.dam.2005.05.010}.
  
  \bibitemdeclare{article}{istrate-clustering}
  \bibitem{istrate-clustering}
  \bibinfo{author}{G.~\surnamestart Istrate\surnameend} (\bibinfo{year}{2007}):
    \emph{\bibinfo{title}{Satisfiability of Boolean Random Constraint
    Satisfaction Problems: Clusters and Overlaps}}.
  \newblock {\slshape \bibinfo{journal}{Journal of Universal Computer Science}}
    \bibinfo{volume}{13}(\bibinfo{number}{11}), pp. \bibinfo{pages}{1655--1670},
    \doi{10.3217/jucs-013-11-1655}.
  
  \bibitemdeclare{article}{aimath04}
  \bibitem{aimath04}
  \bibinfo{author}{G.~\surnamestart Istrate\surnameend},
    \bibinfo{author}{A.~\surnamestart Percus\surnameend} \&
    \bibinfo{author}{S.~\surnamestart Boettcher\surnameend}
    (\bibinfo{year}{2005}): \emph{\bibinfo{title}{Spines of Random Constraint
    Satisfaction Problems: Definition and Connection with Computational
    Complexity}}.
  \newblock {\slshape \bibinfo{journal}{Annals of Mathematics and Artificial
    Intelligence}} \bibinfo{volume}{44}(\bibinfo{number}{4}), pp.
    \bibinfo{pages}{353--372}, \doi{10.1007/s10472-005-7033-2}.
  
  \bibitemdeclare{misc}{github-overlap-ec}
  \bibitem{github-overlap-ec}
  \bibinfo{author}{Gabriel \surnamestart Istrate\surnameend} \&
    \bibinfo{author}{Romeo \surnamestart Negrea\surnameend}
    (\bibinfo{year}{2023}): \emph{\bibinfo{title}{Companion site on Github for
    the paper "$q$-{O}verlaps in the {R}andom {E}xact {C}over {P}roblem" (this paper)}}.
  \newblock \bibinfo{howpublished}{Available at URL
    \url{https://github.com/Gabriel-Istrate/Overlap-Exact-Cover}}.
  
  \bibitemdeclare{article}{kalapala-moore-journal}
  \bibitem{kalapala-moore-journal}
  \bibinfo{author}{V.~\surnamestart Kalapala\surnameend} \&
    \bibinfo{author}{C.~\surnamestart Moore\surnameend} (\bibinfo{year}{2008}):
    \emph{\bibinfo{title}{The phase transition in exact cover}}.
  \newblock {\slshape \bibinfo{journal}{Chicago Journal of Theoretical Computer
    Science}} (\bibinfo{number}{5}), \doi{10.4086/cjtcs.2008.005}.
  
  \bibitemdeclare{techreport}{cs/050837}
  \bibitem{cs/050837}
  \bibinfo{author}{Vamsi \surnamestart Kalapala\surnameend} \&
    \bibinfo{author}{Cris \surnamestart Moore\surnameend} (\bibinfo{year}{2005}):
    \emph{\bibinfo{title}{The Phase Transition in Exact Cover}}.
  \newblock \bibinfo{type}{Technical Report} \bibinfo{number}{cs/0508037},
    \bibinfo{institution}{arXiv.org}, \doi{10.48550/arXiv.cs/0508037}.
  
  \bibitemdeclare{techreport}{knysh-2004}
  \bibitem{knysh-2004}
  \bibinfo{author}{S.~\surnamestart Knysh\surnameend}, \bibinfo{author}{V.~N.
    \surnamestart Smelyanskiy\surnameend} \& \bibinfo{author}{R.~D. \surnamestart
    Morris\surnameend} (\bibinfo{year}{2004}):
    \emph{\bibinfo{title}{Approximating satisfiability transition by suppressing
    fluctuations}}.
  \newblock \bibinfo{type}{Technical Report} \bibinfo{number}{cond-mat/0403416},
    \bibinfo{institution}{arXiv.org}, \doi{10.48550/arXiv.cond-mat/0403416}.
  
  \bibitemdeclare{inproceedings}{zdeborova-2}
  \bibitem{zdeborova-2}
  \bibinfo{author}{E.~\surnamestart Maneva\surnameend},
    \bibinfo{author}{T.~\surnamestart Meltzer\surnameend},
    \bibinfo{author}{J.~\surnamestart Raymond\surnameend},
    \bibinfo{author}{A.~\surnamestart Sportiello\surnameend} \&
    \bibinfo{author}{L.~\surnamestart Zdeborov\'{a}\surnameend}
    (\bibinfo{year}{2007}): \emph{\bibinfo{title}{A hike in the phases of the
    1-in-3 Satisfiability Problem}}.
  \newblock In \bibinfo{editor}{J.P. \surnamestart Bouchaud\surnameend},
    \bibinfo{editor}{M.~\surnamestart M\'{e}zard\surnameend} \&
    \bibinfo{editor}{J.~\surnamestart Dalibard\surnameend}, editors: {\slshape
    \bibinfo{booktitle}{Lecture Notes of the Les Houches Summer School 2006}},
    \bibinfo{publisher}{Elsevier}, pp. \bibinfo{pages}{491--498},
    \doi{10.1016/S0924-8099(07)80022-9}.
  
  \bibitemdeclare{article}{cond-mat/0504070/prl}
  \bibitem{cond-mat/0504070/prl}
  \bibinfo{author}{M.~\surnamestart M\'{e}zard\surnameend},
    \bibinfo{author}{T.~\surnamestart Mora\surnameend} \&
    \bibinfo{author}{R.~\surnamestart Zecchina\surnameend}
    (\bibinfo{year}{2005}): \emph{\bibinfo{title}{Clustering of solutions in the
    random satisfiability problem}}.
  \newblock {\slshape \bibinfo{journal}{Physical Review Letters}}
    \bibinfo{volume}{94}(\bibinfo{number}{197205}),
    \doi{10.1103/PhysRevLett.94.197205}.
  
  \bibitemdeclare{book}{mezo2022lambert}
  \bibitem{mezo2022lambert}
  \bibinfo{author}{Istv{\'a}n \surnamestart Mez{\~o}\surnameend}
    (\bibinfo{year}{2022}): \emph{\bibinfo{title}{The Lambert W function: its
    generalizations and applications}}.
  \newblock \bibinfo{publisher}{CRC Press}, \doi{10.1201/9781003168102}.
  
  \bibitemdeclare{article}{continuous-discontinuous-journal}
  \bibitem{continuous-discontinuous-journal}
  \bibinfo{author}{C.~\surnamestart Moore\surnameend},
    \bibinfo{author}{G.~\surnamestart Istrate\surnameend},
    \bibinfo{author}{D.~\surnamestart Demopoulos\surnameend} \&
    \bibinfo{author}{M.~\surnamestart Vardi\surnameend} (\bibinfo{year}{2007}):
    \emph{\bibinfo{title}{A continuous-discontinuous second-order transition in
    the satisfiability of a class of {H}orn formulas}}.
  \newblock {\slshape \bibinfo{journal}{Random Structures and Algorithms}}
    \bibinfo{volume}{31}(\bibinfo{number}{2}), pp. \bibinfo{pages}{173--185},
    \doi{10.1002/rsa.20176}.
  
  \bibitemdeclare{book}{moore2011nature}
  \bibitem{moore2011nature}
  \bibinfo{author}{C.~\surnamestart Moore\surnameend} \&
    \bibinfo{author}{S.~\surnamestart Mertens\surnameend} (\bibinfo{year}{2011}):
    \emph{\bibinfo{title}{The nature of computation}}.
  \newblock \bibinfo{publisher}{Oxford University Press},
    \doi{10.1093/acprof:oso/9780199233212.001.0001}.
  
  \bibitemdeclare{book}{sfi-book}
  \bibitem{sfi-book}
  \bibinfo{editor}{A.~\surnamestart Percus\surnameend},
    \bibinfo{editor}{G.~\surnamestart Istrate\surnameend} \&
    \bibinfo{editor}{C.~\surnamestart Moore\surnameend}, editors
    (\bibinfo{year}{2006}): \emph{\bibinfo{title}{Computational Complexity and
    Statistical Physics}}.
  \newblock \bibinfo{publisher}{Oxford University Press},
    \doi{10.1093/oso/9780195177374.001.0001}.
  
  \bibitemdeclare{article}{zdeborova-1}
  \bibitem{zdeborova-1}
  \bibinfo{author}{J.~\surnamestart Raymond\surnameend},
    \bibinfo{author}{A.~\surnamestart Sportiello\surnameend} \&
    \bibinfo{author}{L.~\surnamestart Zdeborov\'{a}\surnameend}
    (\bibinfo{year}{2007}): \emph{\bibinfo{title}{The phase diagram of random
    1-in-3 Satisfiability}}.
  \newblock {\slshape \bibinfo{journal}{Phys. Rev. E}}
    \bibinfo{volume}{76}(\bibinfo{number}{011101}),
    \doi{10.1103/PhysRevE.76.011101}.
  
  \bibitemdeclare{article}{phase-transition-random-hypergraph}
  \bibitem{phase-transition-random-hypergraph}
  \bibinfo{author}{J.~\surnamestart Schmidt-{P}ruznan\surnameend} \&
    \bibinfo{author}{D.~\surnamestart Shamir\surnameend} (\bibinfo{year}{1985}):
    \emph{\bibinfo{title}{Component structure in the evolution of random
    hypergraphs}}.
  \newblock {\slshape \bibinfo{journal}{Combinatorica}} \bibinfo{volume}{5}, pp.
    \bibinfo{pages}{81--94}, \doi{10.1007/BF02579445}.
  
  \end{thebibliography}

\end{document}